\documentclass[conference]{IEEEtran}
\usepackage{cite}
\usepackage{amsmath}
\usepackage{enumitem}

\usepackage{graphicx}
\usepackage{subcaption} 
\usepackage{amsthm}
\newtheorem{theorem}{Theorem}  
\newtheorem{lemma}[theorem]{Lemma}      

\newtheorem{proposition}[theorem]{Proposition}
\newtheorem{definition}{Definition}
\newtheorem{example}{Example}

\usepackage{comment}

\usepackage{fmtcount}
\usepackage{amssymb}
\begin{document}
\title{Entropy Functions on Two-Dimensional Faces of Polymatroid Region Spanned by a Matroid and a Rank-One Matroid}
\author{Kaizhe He and Qi Chen \\
	\IEEEauthorblockA{\textit{School of Telecommunications Engineering} \\
		\textit{Xidian University}\\
		Shaanxi, China \\
		kz.he@stu.xidian.edu.cn, qichen@xidian.edu.cn}
	}

\maketitle
\begin{abstract}
Characterization of entropy functions is of fundamental importance in information theory. By imposing constraints on their Shannon outer bound, i.e., the polymatroidal region, one obtains the faces of the region and entropy functions on them with special structures. In this paper, we characterize entropy functions on 2-dimensional faces of polymatroidal region of degree n spanned by a matroid and a rank-1 matroid. We classify all such 2-dimensional faces into four types.
\end{abstract}

\section{Introduction}
Let $N_n=\{1,2,...,n\}$ and $(X_i,i\in N_n)$ be a random vector with each $X_i$ distributed on $\mathcal{X}_i$. The entropy function of $(X_i,i\in N_n)$  is defined as a set function $\mathbf{h}: 2^{N_n} \to \mathbb{R}$ with $\mathbf{h}(A) = H(X_A)$ for each $A \subseteq N_n$. The Euclidean space $\mathcal{H}_n \triangleq \mathbb{R}^{2^{N_n}}$ in which such functions reside is referred to as the $\mathit{entropy}$ $\mathit{space}$ of degree $n$. The collection of all entropy functions, denoted by $\Gamma_n^*$, is called the $\mathit{entropy}$ $\mathit{region}$. 

As an outer bound of $\Gamma_n^*$, the polymatroidal region $\Gamma_n$ is defined by a set of linear constraints known as the polymatroid axioms: for all $A,B\subseteq N_n$,
\begin{align}
\mathbf{h}(A)&\geq 0,\label{eq:information1}\\
\mathbf{h}(A)&\leq \mathbf{h}(B)\ \text{if}\ A\subseteq B,\label{eq:information2}\\ 
\mathbf{h}(A)+\mathbf{h}(B)&\geq \mathbf{h}(A\cap B)+\mathbf{h}(A\cup B).\label{eq:information3}
\end{align}
 Any $\mathbf{h}\in \Gamma_n$ is called (the rank function of) a $\mathit{polymatroid}$. These axioms are equivalent to the nonnegativity of Shannon information measures. Thus, $\Gamma_n$ can also be viewed as the polyhedral cone determined by Shannon-type information inequalities. In 1998, Zhang and Yeung discovered the first non-Shonnon type ineqality for four random variables \cite{zhang1998characterization}, thus $\overline{\Gamma_n^*}\subsetneq\Gamma_n$ when $n \geq 4$, where $\overline{\Gamma_n^*}$ is the closure of $\Gamma_n^*$. In 2006, Mat\'{u}\v{s} proved that $\overline{\Gamma_n^*}$ is not polyhedral for $n\geq 4$, that is, there exists infinitely many independent information inequalities for a fixed $n\geq 4$\cite{matus2007two}. The characterization of $\Gamma_n^*$ and its closure $\overline{\Gamma_n^*}$ remains a challenging open problem.

In this paper, we characterize $\Gamma_n^*$ by determining its intersection with a face of $\Gamma_n$. (See the definition of the faces of a polyhedral cone in Section \ref{2A}.) For a face $F$ of $\Gamma_n$, we say characterizing $F$ for short for determining $F^*\triangleq F\cap{\Gamma_n^*}$. In 2006, Mat\'{u}\v{s} characteried the first non-trivial 2-dimensional face of $\Gamma_3$ and proved a piecewise linear conditional inequality for entropy functions on it\cite{matus2005piecewise}. In 2012, Chen and Yeung characterized another type of 2-dimensional face of $\Gamma_3$\cite{chen2012characterizing}. Recently, Liu and Chen  systematically enumerated all 59 types of 2-dimensional faces of $\Gamma_4$ and characterized almost all of them\cite{liu2023entropy}\cite{liu2025entropy}. 


Building upon these work, this paper continues to investigate entropy functions on  2-dimensional faces of $\Gamma_n$ for general $n$. We characterize 2-dimensional faces with one extreme ray containing rank-1 matroid, and classify all such 2-dimensional faces with another extreme ray containing a matroid into four types, i.e., the all-entropic, Mat\'{u}\v{s}-type, Chen-Yeung-type and the non-entropic.

We formulate the problem in Section \ref{2}, where Section \ref{2A} and \ref{2B} give the preliminaries on matroids and polyhedral cones, respectively, and the definitions of the four types of 2-dimensional faces are given in Section \ref{2C}. The classfictions of these 2-dimensional faces and the characterization of entropy fynctions on them are in Section \ref{3}.
\section{Problem formulation}
\label{2}
\subsection{Preliminaries on matroids}
\label{2A}
For a polymatroid $P=(N_n,\mathbf{h})$, $N_n$ is called its $\mathit{ground}$ $\mathit{set}$ and $\mathbf{h}\in\Gamma_n$ is called its $\mathit{rank}$ $\mathit{functions}$. A polymatroid is called $\mathit{integer}$ if its rank function takes only integer values, i.e., $\mathbf{h}(A) \in \mathbb{Z}$ for all $A \subseteq N_n$. When an integer polymatroid further satisfies $\mathbf{h}(A) \leq |A|$ for all $A \subseteq N_n$, it reduces to a $\mathit{matroid}$ $M=(N_n,\mathbf{r})$. A matroid is called $\mathit{uniform}$ and denoted by $U_{k,n}$ if its rank function is defined by \(\mathbf{r}(A) = \min\{k, |A|\}\) for all \(A \subseteq N_n\). 

For a matroid $M=(N_n,\mathbf{r})$, an element $e\in N_n$ is a $\mathit{loop}$ if $\mathbf{r}(\{e\}) = 0$. Two distinct nonloops $e$ and $e'\in N_n$ are $\mathit{parallel}$ if $\mathbf{r}(\{e\}) = \mathbf{r}(\{e'\}) = \mathbf{r}(\{e, e'\}) = 1$. For $X\subseteq N$, it is a $\mathit{circuit}$ of $M$ if for any $x \in X$, $\mathbf{r}(X-x)=|X|-1=\mathbf{r}(X)$. Then $M$ is $\mathit{connected}$ if for every pair of distinct elements $e,f\in N_n$, there exists a circuit containing both $e$ and $f$. Note that when $M=U_{n-1,n}$, all elements in the ground set $N_n$ form the unique circuit of the matroid.  

For a positive integer $k\leq n'\leq n$ and $\alpha\subseteq N_n$ with $|\alpha|=n'$, we define the matroid $U_{k,n'}^{\alpha,n}$ whose restriction on $\alpha$ is a uniform matroid $U_{k,n'}$, and each $e\in N_n\backslash \alpha$ is a loop. When $k=1$, the matroid $U_{1,n'}^{\alpha,n}$ specifies a partition of $N_n = \alpha \cup L$, where
\begin{itemize}
    \item $\alpha$ is a subset of size $n'$ whose elements are pairwise parallel, and
    \item $L=N_n\backslash \alpha$ is the set of all loops of $U_{1,n'}^n$.
\end{itemize}
If $M$ is a rank 1 matroid, then $\alpha$ is a set of elements with rank $1$. It can be seen $\mathbf{r}(X)=1$ if $X$ intersects $\alpha$, and $\mathbf{r}(X)=0$ otherwise. Write $U^n_{1,n'}$ for $\alpha=N_{n'}$.

For more about matroid theory, readers can refer to \cite{oxley2011matroid}.
\subsection{Preliminaries on polyhedral cones}\label{2B}
By definition, $\Gamma_n$ is a $\mathit{polyhedral}$ $\mathit{cone}$, that is, a $\mathit{cone}$ determined by a finite number of linear inequalities. For a polyhedral cone  $C \subseteq \mathbb{R}^d $, if there exists a hyperplane $ P $ such that $C$ is contained in one of the closed halfspaces determined by $ P $ and $ C \cap P \neq \varnothing $, then $ C \cap P $ is called a $\mathit{face}$ of $ C $. In particular, a face of dimension 1 is called an $\mathit{extreme}$ $\mathit{ray}$, and a face of dimension $ d-1 $ is called a $\mathit{facet}$. Each face can be equivalently described either as the set of nonnegative combinations of its extreme rays (V-representation) or as the intersection of all facets containing it (H-representation). For more about polyhedral cones, please consult \cite{ziegler2012lectures}.

For an extreme ray of $\Gamma_n$, it contains an integer polymatroid because all facets of $\Gamma_n$ have integer coefficients. For an integer polymatroid $P$ contained by an extreme ray, it is called $\mathit{minimal}$ if $\frac{1}{t}\mathbf{r}_P$ is integer only when $t=1$ for any positive integer $t$. In this paper, extreme rays are identified with the minimal integer polymatroid contained in the ray.

Note that all 2-dimensional faces of $\Gamma_n$ are spanned by 2 extreme rays. Thus, a 2‑dimensional face is denoted by $(P_1, P_2)$, where $P_1$ and $P_2$ are two distinct extreme rays on the face.

Note that polymatroid axioms are equivalent to the elemental inequalities\eqref{eq:information1}-\eqref{eq:information3},
\begin{align}
\mathbf{h}({N_n})&\geq\mathbf{h}({N_n\backslash i}),\qquad i\in N_n;\label{eq:u1.1}\\
\mathbf{h}(K)+\mathbf{h}(K\cup ij)&\leq\mathbf{h}(K\cup i)+\mathbf{h}(K\cup j),\label{eq:u1.2}\nonumber\\
i,j\in N_n,K&\subseteq N_n\backslash\{i,j\}
\end{align}
each of which determines a facet of $\Gamma_n$\cite[Chapter 14]{yeung2008information}. 
Let
\begin{align}
F(i)=\{\mathbf{h}&\in\Gamma_n:\mathbf{h}(N_n)=\mathbf{h}(N\backslash i)\},\label{eq:f1}\\
F(i;j|K)=\{\mathbf{h}&\in\Gamma_n:\mathbf{h}(K)+\mathbf{h}(K\cup ij)\label{eq:f2}\nonumber\\
&=\mathbf{h}(K\cup i)+\mathbf{h}(K\cup j)\},
\end{align}
be the two types of facets determined by \eqref{eq:u1.1} and \eqref{eq:u1.2}, respectively.

For $\mathbf{h}\in\Gamma_n$, it is called $\mathit{modular}$ if $\mathbf{h}(A)=\sum\limits_{i\in A}\mathbf{h}(i)$. Note that the family of all modular polymatroids forms a face $F^\textup{mod}$ of $\Gamma_n$, and we call $F^\textup{mod}$ the $\mathit{modular}$ $\mathit{face}$ of $\Gamma_n$. It can be checked that
\begin{align}
	F^\textup{mod}=\bigcap\limits_{i;j|K}F(i;j|K)=\textup{cone}(\mathbf{r}_{U_{1,1}^{k,n}},k\in N_n),\nonumber
\end{align}  
 which are H-representation and V-representation of $F^\textup{mod}$, respectively. 

For $\mathbf{h}\in\Gamma_n$, it is called $\mathit{tight}$ if $\mathbf{h}(N_n)=\mathbf{h}(N_n\backslash i)$. Note that the family of all tight polymatroids forms a face $F^\textup{ti}$ of $\Gamma_n$, and we call $F^\textup{ti}$ the $\mathit{tight}$ $\mathit{face}$ of $\Gamma_n$. It can be checked that
\begin{align}
	F^\textup{ti}=\bigcap\limits_{i\in N_n}F(i)=\textup{cone}(\mathbf{r}_P,P\neq {U_{1,1}^{k,n}},k\in N_n), \nonumber
\end{align}
which are H-representation and V-representation of $F^\textup{ti}$, respectively. 

\subsection{Four types of ($P,U_{1,n'}^n$)}\label{2C}

For an integer polymatroid $P$, we define the probabilistically (p-)characteristic set of $P$ as
\[
\chi_	P\triangleq \{v\in \mathbb{Z}:v\geq2, \log v\cdot \mathbf{r}_P\in\Gamma_n^*\}.
\]
According to \cite{chen2021matroidal}\cite{chen2024matroidal}, the p-characteristic set of $P$ characterizes extreme rays containing $P$ when $P$ is a connected matroid with rank exceeding 1.

We embed each face $F=(P,U_{1,{n'}}^n)$ in the first quadrant of the 2-dimensional cartesian coordinate system whose axes are labeled by $a$ and $b$. The vector $(a,b)$ represents the polymatroid $a\mathbf{r}_P + b\mathbf{r}_{U_{1,n'}^n}$, where $\mathbf{r}_P$ is the rank function of the minimal integer polymatroid $P$ in an extreme ray of $F$ and $\mathbf{r}_{U_{1,n'}^n}$ is the rank function of $U_{1,{n'}}^n$, respectively. We define four types of $F$ according to $F^*\triangleq F\cap\Gamma_n^*$, and will prove in Section \ref{3}. That, all $F_S$ can be classified into these four types when $P$ is a matroid.
\begin{definition}\label{d1}
A 2-dimensional face ($P,U_{1,{n'}}^n$) of $\Gamma_n$ is called
\begin{enumerate}   
    \item \textbf{all-entropic }if ($a,b$) is entropic for all $a,b>0$;
    \item \textbf{Mat\'{u}\v{s}-type }if ($a,b$) is entropic when $a+b\geq\log v$ and $\log (v-1)<a<\log v$ for positive integer $v\in \chi_P$ and non-entropic when $a+b<\log \lceil e^a \rceil$; 
    \item \textbf{Chen-Yeung-type }if ($a,b$) is entropic when $a=\log v,b>0$ for positive integer $v\in \chi_P$ and non-entropic when $a\neq\log v,b>0$ for positive integer $v$;
    \item \textbf{non-entropic }if ($a,b$) is non-entropic for all $a>0$, $b>0$.
\end{enumerate}
\end{definition}
\indent By Lemma \ref{h1h2} in Section \ref{3}, $(P,U_{1,n'}^n)$ is non-entropic only when $P$ is non-entropic, that is, each $\mathbf{h}\in P$ is non-entropic except for the origin.

By Lemma \ref{all} in Section \ref{3c}, We can determine easily whether a 2-dimensional face is all-entropic. Thus, our discussion will be focused mainly on Chen–Yeung-type and Mat\'{u}\v{s}-type.
\begin{example}
The first characterized Mat\'{u}\v{s}-type and Chen-Yeung-type faces are $(U_{2,3},U_{1,2}^3)$ in Fig.\ref{fig:1a} and $(U_{2,3},U_{1,1}^3)$ in Fig.\ref{fig:1b}, respectively\textup{\cite{matus2005piecewise}\cite{chen2012characterizing}}. The faces $(U_{2,4},U_{1,3}^4)$ in Figure \ref{fig:1c} and $(U_{2,4},U_{1,2}^4)$ in Figure\ref{fig:1d} are, respectively, Mat\'{u}\v{s}-type and Chen-Yeung-type as well \textup{\cite{liu2025entropy}}. Note that as $\chi_{U_{2,4}}=\{v\in\mathbb{Z},v\geq 3, v\neq 6\}$, there exist some missing pieces for $v=2$ and $6$ in the two faces.
\begin{figure}[htbp]
    \centering
    \begin{subfigure}[b]{0.23\textwidth}
        \centering
        \includegraphics[width=\linewidth]{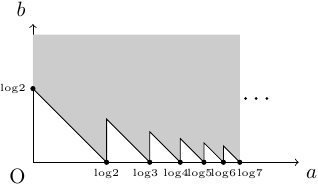}
        \caption{the face $(U_{2,3},U_{1,2}^3)$}
        \label{fig:1a}
    \end{subfigure}
    \hfill
    \begin{subfigure}[b]{0.23\textwidth}
        \centering
        \includegraphics[width=\linewidth]{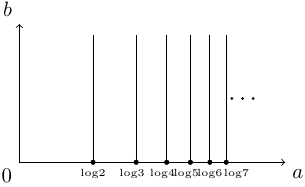}
        \caption{the face $(U_{2,3},U_{1,1}^3)$}
        \label{fig:1b}
    \end{subfigure}
    
    \vspace{0.1cm} 
    
    \begin{subfigure}[b]{0.23\textwidth}
        \centering
        \includegraphics[width=\linewidth]{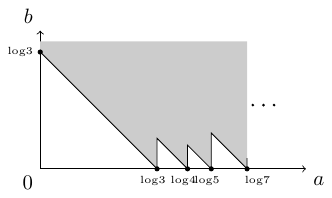}
        \caption{the face $(U_{2,4},U_{1,3}^4)$}
        \label{fig:1c}
    \end{subfigure}
    \hfill
    \begin{subfigure}[b]{0.23\textwidth}
        \centering
        \includegraphics[width=\linewidth]{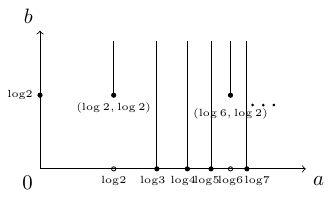}
        \caption{the face $(U_{2,4},U_{1,2}^4)$}
        \label{fig:1d}
    \end{subfigure}
    
    \caption{Chen-Yeung-type and Mat\'{u}\v{s}-type faces}
    \label{fig:bigfigure}
\end{figure}
\end{example}

\section{Entropy functions on 2-dimensional faces $(M,U_{1,n'}^n)$}\label{3}
\subsection{Entropy functions on $(M,U_{1,1}^n)$} 

\begin{lemma}\label{modti}
 If $\mathbf{h}$ is in an extreme of $\Gamma_n$, then, $\mathbf{h}\in F^\textup{mod}$ or $\mathbf{h}\in F^\textup{ti}$.
\end{lemma}
Lemma \ref{modti} can be directedly obtained from the V-representions of both $F^\textup{mod}$ and $F^\textup{ti}$, and also \cite[Section III]{matuvs2016entropy}.
\begin{proposition}\label{u11 face}
For any extreme ray ${P}$ of ${\Gamma_n}$, $({P},U_{1,1}^n)$ is a 2-dimensional face of $\Gamma_n$.
\end{proposition}
\begin{proof} 
By Lemma \ref{modti}, $P$ is either modular or tight.

If $P$ is modular, then $P$ is $U_{1,1}^{k,n}$ with $k\neq 1$. Then it can be checked that $(U_{1,1}^{k,n},U_{1,1}^{n})$ is a 2-dimensional face of $\Gamma_n$ contained in all facets except for $F(1)$ and $F(k)$.

If $P$ is tight, assume the contrary that $P$ and $U_{1,1}^n$ do not span a 2-dimensional face. Then it contains an extreme ray $P'\neq P,U_{1,1}^n$.
\begin{enumerate}
	\item If $P'$ is modular, then $P'=U_{1,1}^{k,n}$ with $k\neq 1$. Note that $F(k)$ do not contain $U_{1,1}^{k,n}$ but contain both $P$ and $U_{1,1}^n$, a contradiction.
	\item If $P'$ is tight, then $P'$ is contained by all $F(i)$. Thus $P'$ is contained by all $F(i)$ and $F(i;j|K)$ that contain $P$, which implies $P'\neq P$, a contradiction as well.
\end{enumerate}

	Hence, $(P,U_{1,1}^n)$ is a 2-dimensional face of $\Gamma_n$.
\end{proof}

\begin{lemma}\label{h1h2}\textup{\cite[Lemma 15.3]{yeung2008information}}
For any $\mathbf{h}_1$, $\mathbf{h}_2\in \Gamma_n^*$, $\mathbf{h}_1+\mathbf{h}_2\in \Gamma_n^*$.
\end{lemma}
\begin{lemma}\label{uniform}\textup{\cite[Lemma 3]{liu2023entropy}}
If $X_1$ and $X_2$ are independent and for any $p(x_1,x_2,x_3)>0$, $p(x_1)=p(x_2)$, then $X_1$ and $X_2$ are uniformly distributed on $\mathcal{X}_1$ and $\mathcal{X}_2$, respectively, $|\mathcal{X}_1|=|\mathcal{X}_2|$ and $H(X_1)=H(X_2)$.
\end{lemma}
\begin{lemma}\textup{\cite[Theorem 2.1.5]{nguyen1978semimodular}}\label{connected}
A matroid is in an extreme ray of $\Gamma_n$ if and only if it is connected by deleting its loops.
\end{lemma}

\begin{lemma}\label{u11}
For $F=(U_{n-1,n},U_{1,1}^n)$ and $n\ge 3$, $\mathbf{h}=(a,b)\in F$ is entropic if and only if $a=\log v$ for some positive integer $v$, that is, $F$ is Chen-Yeung-type.
\end{lemma}
\begin{proof}
For $F=(U_{n-1,n},U_{1,1}^n)$, if $\mathbf{h} \in F$ is entropic, its characterizing random vector ($X_i \in  N_n$) satisfies the following information equalities,
\begin{align}
H(X_{N_n})&=H(X_{N_{n}-i}),i \in N_n,i\neq 1,\label{eq:u11.1} \\
H(X_{i \cup K}) + H(X_{j \cup K}) &= H(X_{K}) + H(X_{ij \cup K}).\label{eq:u11.2}\nonumber \\
K\subseteq N_n,|K|\le &n-3,i,j\in N_n     \backslash K
\end{align}
For $(x_i, i \in N_n) \in \mathcal{X}_{N_n}$, with $p(x_{N_n}) > 0$, above information equalities imply that the probability mass function satisfies
\begin{align}
p(x_{N_n})&=p(x_{N_n}-i),i\in N_n,i\neq 1,\label{eq:nn}\\
p(x_ix_K)p(x_jx_K)&=p(x_K)p(x_ix_jx_K).\label{eq:ijk} \nonumber \\
K\subseteq N_n,|K|\le &n-3,i,j\in N_n     \backslash K
\end{align}
Routine calculation leads to
\begin{align}
p(x_1x_2)&=p(x_1x_3),\label{eq:u11.4}\\
p(x_1x_2)&=p(x_1)p(x_2),\label{eq:u11.5}\\
p(x_1x_3)&=p(x_1)p(x_3).\label{eq:u11.6}
\end{align}
Then, we can get
\begin{align}
p(x_2)=p(x_3).\label{eq:u11.7}
\end{align}
Note that $X_2$ and $X_3$ are independent, by Lemma \ref{uniform}, $X_2$ and $X_3$ are uniformly distributed on $\mathcal{X}_2$ and $\mathcal{X}_3$, respectively, and so $H(X_2) = H(X_3) = \log v$ where $v = |\mathcal{X}_2| = |\mathcal{X}_3|$.\\
\indent The "only if" part is immediately implied by Lemma \ref{h1h2} and the fact that $a = \log v$ on the ray $U_{n-1,n}$, and the whole ray $U_{1,1}$ is entropic.
\end{proof}

\begin{theorem}\label{u11cy}
For any matroid ${M}$ of rank $r\geq2$ in an extreme ray of $\Gamma_n$, $({M},U_{1,1}^n)$ is Chen-Yeung-type.
\end{theorem}
\begin{proof}
We prove the theorem in the cases that $1\in N_n$ is either a loop of $M$ or not.
\begin{enumerate}
	\item If $1\in N_n$ is a loop of $M$, let $i,j\in N_n$ that are neither parallel nor loops, that is, $\mathbf{r}_M(i,j)=2$. Then, by Lemma \ref{connected}, there must be a cycle $C$ containing $\{i,j\}$ with $|C|=k\geq 3$. For $\mathbf{h}=(a,b)\in F$, $\mathbf{h}(A)=a\mathbf{r}_M(A)+b\mathbf{r}_{U_{1,1}^n}(A)$. For all $A\subseteq C$, as $1\notin C$, $\mathbf{r}_{U_{1,1}^n}(A)=0$. So restricting $\mathbf{h}$ on $C$, we obtain $\mathbf{h}_C$ is the rank function of a $U_{k-1,k}$ on $C$. Thus by \cite[Proposition 1]{chen2021matroidal}, if $\mathbf{h}$ is entropic, $a=\log v$ for positive integer $v$. 
	\item If $1\in N_n$ is not a loop of $M$, let $i\in N_n$ that is neither parallel with $1$ nor loop, that is, $\mathbf{r}_M(1,j)=2$. Then, by Lemma \ref{connected}, there must be a cycle $C$ containing $\{1,i\}$ with $|C|=k\geq 3$. For $\mathbf{h}=(a,b)\in F$, $\mathbf{h}=a\mathbf{r}_M+b\mathbf{r}_{U_{1,1}^n}$. Restricting $\mathbf{h}$ on $C$, we obtain $\mathbf{h}_C=a\mathbf{r}_1+b\mathbf{r}_2$, where $r_i$, $i=1,2$ are the rank functions of $U_{k-1,k}$ and $U_{1,1}^n$ on $C$, respectively. Thus, by Lemma \ref{u11}, if $\mathbf{h}$ is entropic, $a=\log v$ for positive integer $v$. 
\end{enumerate}

Hence $\{(a,b): a=\log v,b\geq 0, v\in\mathbb{Z}\}$ forms an outer bound on $F^*$. The inner bound $\{(a,b):a=\log v,b\geq 0, v\in \chi_M$\} is immediately implied be Lemma \ref{h1h2} and $\chi_M$. Hence, $F$ is Chen-Yeung-type.
\end{proof}
\addtolength{\rightmargin}{0.03in}
\subsection{Entropy functions on $(M,U_{1,2}^n)$}

For any matroid $M$ in an extreme ray of $\Gamma_n$, $(M,U_{1,2}^n)$ is a 2-dimensional face. 
\begin{lemma}\label{circuit not contain}
	For any two distinct matroids $M_1$ and $M_2$ in distinct extreme rays of $\Gamma_n$, let $\mathcal{C}_1$ and $\mathcal{C}_2$ be the family of circuits of $M_1$ and $M_2$, respectively. Then $\mathcal{C}_1 \not\subseteq \mathcal{C}_2$ and $\mathcal{C}_2 \not\subseteq \mathcal{C}_1$.
\end{lemma}
\begin{proof}
	We prove the lemma by contradiction. Assume $\mathcal{C}_1\subseteq\mathcal{C}_2$, then there exists $C\in\mathcal{C}_2\backslash\mathcal{C}_1$ and $\mathbf{r}_{M_2}(C)=|C|-1$. If $C$ is dependent in $M_1$, then there exists $C'\subseteq C$ such that $C'\in\mathcal{C}_1$, which implies $C'\in \mathcal{C}_2$ and both $C$ and $C'\in\mathcal{C}_2$, contradicting to \cite[Corallary 1.1.5]{oxley2011matroid} which claims that two distinct circuits of a matroid can not contain each other. Then $C$ is independent in $M_1$ and so $\mathbf{r}_{M_1}(C)=|C|$.

For any $i\in N_n\backslash C$, by \cite[Lemma 1.4.2, Proposition 1.4.11]{oxley2011matroid}, we obtain
\begin{align}
\mathbf{r}(C\cup i)-\mathbf{r}(C)= 
\begin{cases}
0, & \text{if there exists circuit $C'$ such }\\
   &	\text{that $C'\not\subseteq C$ but $C'\subseteq C+i$,}\\
1, & \text{otherwise,}
\end{cases}
\end{align}
then, we can get the following inequalities by the fact that $\mathcal{C}_1\subseteq\mathcal{C}_2$,
\begin{align}
\mathbf{r}_{M_1}(C\cup i)-\mathbf{r}_{M_1}(C)\geq\mathbf{r}_{M_2}(C\cup i)-\mathbf{r}_{M_2}(C).
\end{align}
Furthermore, for $S\subseteq N_n\backslash C$, we obtain
\begin{align}\label{eq:guodu}
\mathbf{r}_{M_1}(C\cup S)-\mathbf{r}_{M_1}(C)\geq\mathbf{r}_{M_2}(C\cup S)-\mathbf{r}_{M_2}(C).
\end{align}
Assume $\mathbf{r}_{M_2}(N_n)=k$. Substituting $S=N_n\backslash C$ into inequality \eqref{eq:guodu}, we obtain
\begin{align}
\mathbf{r}_{M_1}(N_n)-(|C|-1)\geq k-|C|.
\end{align}
Thus, $\mathbf{r}_{M_1}(N_n)=k'\geq k+1$ and there exist a circuit $C'$ of length $k'+1$ in $M_1$. Obviously, the maximum length of circuits in $M_2$ is $k+1\leq k'+1$, thus $C'$ is not in $M_2$.
Therefore, $\mathcal{C}_1\not\subseteq\mathcal{C}_2$ and $\mathcal{C}_2\not\subseteq\mathcal{C}_1$.
\end{proof}

\begin{proposition}\label{u12face}
	For any matroid $M$ in an extreme ray of $\Gamma_n$, $(M,U_{1,2}^n)$ is a 2-dimensional face.
\end{proposition}
\begin{proof}
	$U_{1,2}$ is contained in all facets except for $F(1;2|K)$, $K\subseteq N_n\backslash \{1,2\}$.
	We prove the proposition by contradiction. Assume that the minimal face containing both $M$ and $U_{1,2}^n$ also contains another extreme ray $M'\neq M,U_{1,2}^n$. By Lemma \ref{circuit not contain}, there exists $C\subseteq N_n$ such that $C$ is a circuit of $M'$ but not of $M$.

Assume there exists $C\neq\{1,2\}$ a circuit of $M'$ but not of $M$, it can be checked that for all $x,y\in C$, $\{x,y\}\neq\{1,2\}$, $F(x;y|C\backslash\{x,y\})$ does not contain $M'$. If $C$ is independent in $M$, for any $x,y\in C$, $\{x,y\}\neq\{1,2\}$, $F(x;y|C\backslash\{x,y\})$ contains $M,U_{1,2}^n$, a contradiction. If $C$ is dependent in $M$, that is there exist $C'\subseteq C$ a circuit of $M$, it can be checked that for $x\in C'$, $y\in C\backslash C'$, $\{x,y\}\neq\{1,2\}$, $F(x;y|C\backslash\{x,y\})$ contains $M,U_{1,2}^n$, a contradiction as well.

Assume only $C =\{1,2\}$ is a circuit of $M'$ but not of $M$, that is $1$ and $2$ are parallel in $M'$. Since $M'\neq U_{1,2}^n$, there exists a circuit $C'$ of length $k+1$ of $M'$ and $M$ with $1\in C'$. Then there exists another circuit $C''$ of $M'$ such that $C''=(C'\backslash\{1\})\cup\{2\}$ and  $C''$ is also a circuit of $M$. Since $C=\{1,2\}$ is not a circuit of $M$ but $C'$ and $C''$ are circuits of $M$, by \cite[Lemma 1.1.3]{oxley2011matroid}, $(C''\cup{1})\backslash x$ with $x\in C''\backslash\{2\}$ is $U_{k,k+1}$ in $M$ and $C''\cup\{1\}$ is $U_{k,k+2}$ in $M$. Then it can be checked that $F(x;y|C''\cup \{1\}\backslash\{x,y\})$, $x,y\in C''\backslash\{2\}$ contain both $M$ and $U_{1,2}^n$ but not $M'$, a contradiction.

Hence, $(M,U_{1,2}^n)$ is a 2-dimensional face of $\Gamma_n$.
	
\end{proof}

\begin{lemma}\label{u12}
For $F=(U_{n-1,n},U_{1,k}^n)$ with $1<k\leq n$ and $n\geq 4$, $\mathbf{h}=(a,b)\in F$ is entropic if and only if $a=\log v$ for positive integer $v$, that is, $F$ is Chen-Yeung-type. 
\end{lemma}

\begin{proof}

If $\mathbf{h}\in F$ is entropic, its characterizing random vector($X_i,i\in N_n$) satisfies the following information equalities,
\begin{align}
H(X_{N_n})&=H(X_{N_{n}-i}),i\in N_k \label{eq:u12.1} \\
H(X_{i\cup K})+H(X_{j\cup K})&=H(X_K)+H(X_{ij\cup K}).\label{eq:u12.3}\nonumber \\ 
|K|\geq n-3,&i,j\in N_n\backslash K\text{ and }\nonumber\\
K \not\subseteq N_n \setminus N_k &\text{ or } i\notin N_k \text{ or }j \notin N_k
\end{align}
For $(x_i,i\in N_n)\in \mathcal{X}_{N_n}$, with $p(x_{N_n})>0$, above information equalities imply the probability mass function satisfies
\begin{align}
p(x_{N_n})&=p(x_{N_{n}-i}),i\in N_k  \\
p(x_ix_K)p(X_jx_K)&=p(x_K)p(x_ix_jx_K).\nonumber \\ 
|K|\geq n-3&,i,j\in N_n\backslash K\text{ and }\nonumber\\
K \not\subseteq N_n \setminus N_k &\text{ or } i\notin N_k \text{ or }j \notin N_k
\end{align}
Routine calculation leads to
\begin{align}
p(x_{1(n-1)})&=p(x_{1n}), \label{eq:u12.5} \\
p(x_{1(n-1)})&=p(x_1)p(x_{n-1}), \label{eq:u12.6} \\
p(x_{1n})&=p(x_1)p(x_n).\label{eq:u12.7}
\end{align}
Then, we can get
\begin{align}
p(x_{n-1})=p(x_n).\label{eq:u12.8}
\end{align}
Note that $X_{n-1}$ and $X_n$ are independent. By Lemma \ref{uniform}, $X_{n-1}$ and $X_n$ are uniformly distributed on $\mathcal{X}_{n-1} $ and $\mathcal{X}_n$, respectively, and so $H(X_{n-1})=H(X_n)=\log v$ where $v=|\mathcal{X}_{n-1}|=|\mathcal{X}_n|$.

The "only if" part is immediately implied by Lemma \ref{h1h2} and the fact that $a = \log v$ on the ray $U_{n-1,n}$, and the whole ray $U_{1,k}$ is entropic.
\end{proof}

\begin{theorem}
For any matroid $M$ with rank $r\geq 2$ in an extreme ray of $\Gamma_n$, ($M,U_{1,2}^n$) is 
\begin{itemize}
	\item Mat\'{u}\v{s}-type if $\mathbf{r}_M(N_n)=2$ and there exists a cycle $C=\{1,2,i\}$, such that, $\forall j\in N_n\backslash C$ in $M$, $\{j\}$ is either parallel with $\{i\}$ or a loop;
	\item Chen-Yeung-type, otherwise.
\end{itemize}
\end{theorem}

\begin{proof}
We prove the theorem in the cases $1,2\in N_n$ are parallel; at least one of 1,2 is a loop or $\mathbf{r}_M(1,2)=2$. 
\begin{enumerate}
\item If $\{1,2\}\in N_n$ are parallel of $M$, that is, $\mathbf{r}_M(1)=\mathbf{r}_M(2)=\mathbf{r}_M(1,2)=1$. Then, by Lemma \ref{connected}, a matroid $M_1=M\backslash\{1\}$ is an extreme ray of $\Gamma_{n-1}$. For $\mathbf{h}=(a,b)\in F$, $\mathbf{h}=a\mathbf{r}_M+b\mathbf{r}_{U_{1,2}^n}$. Restricting $\mathbf{h}$ on $N_n\backslash \{1\}$, we obtain $\mathbf{h}_{M_1}=a\mathbf{r}_1+b\mathbf{r}_2$, where $\mathbf{r}_i$, $i=1,2$ are the rank functions of a matroid $M_1$ and $U_{1,1}^{n-1}$ on $M_1$, respectively. Thus, by Theorem \ref{u11cy}, if $\mathbf{h}$ is entropic, $a=\log v$ for positive integer $v$, that is, $F$ is Chen-Yeung-type.
\item If at least one of $\{1,2\}\in N_n$ is a loop of $M$, that is, $\mathbf{r}_M(1)=0$ or $\mathbf{r}_M(2)=0$. WLOG, assume $\mathbf{r}_M(1)=0$. Then, by Lemma \ref{connected}, there must be a matroid $M_2=M\backslash\{1\}$. For $\mathbf{h}=(a,b)\in F$, $\mathbf{h}=a\mathbf{r}_M+b\mathbf{r}_{U_{1,2}^n}$. Restricting $\mathbf{h}$ on $
N_n\backslash \{1\}$, we obtain $\mathbf{h}_{M_2}=a\mathbf{r}_1+b\mathbf{r}_2$, where $\mathbf{r}_i$, $i=1,2$ are the rank functions of a matroid $M_2$ and $U_{1,1}^{n-1}$ on $N_n\backslash\{1\}$, respectively. Thus, by Theorem \ref{u11cy}, if $\mathbf{h}$ is entropic, $a=\log v$ for positive integer $v$, that is, $F$ is Chen-Yeung-type.
\item If $\{1,2\}\in N_n$ are not parallels or loops of $M$, that is, $\mathbf{r}_M(1,2)=2$. Then, by Lemma \ref{connected}, there must be a cycle $C$ containing $\{1,2\}$ with $|C|=k\geq 3$.
\begin{enumerate}
\item When $\mathbf{r}_M(N_n)>2$, by Lemma \ref{connected}, there must be a cycle $C_1$ with $|C_1|\geq 4$. For $\mathbf{h}=(a,b)\in F$, $\mathbf{h}=a\mathbf{r}_M+b\mathbf{r}_{U_{1,2}^n}$. Restricting $\mathbf{h}$ on $C_1$, we obtain $\mathbf{h}_{C_1}=a\mathbf{r}_1+b\mathbf{r}_2$, where $\mathbf{r}_i$, $i=1,2$ are the rank functions of $U_{k-1,k}$ and $U_{1,2}^{n-1}$ on $C_1$, respectively. Thus, by Lemma \ref{u12}, if $\mathbf{h}$ is entropic, $a=\log v$ for positive integer $v$. 
\item When $\mathbf{r}_M(N_n)=2$, $|C|=3$. We assume the cycle $C=\{1,2,i\}$, $i\in N_n\backslash\{1,2\}$.
\begin{enumerate}
\item If $\forall j\in N_n\backslash\{1,2,i\}$ is either parallel with $\{i\}$ or a loop. For $\mathbf{h}=(a,b)\in F$, $\mathbf{h}=a\mathbf{r}_{M}+b\mathbf{r}_{U_{1,2}^n}$, and we can check that $\mathbf{h}$ is the rank function of $(U_{2,3},U_{1,2})$. Thus, by \cite{matus2005piecewise}, $F$ is Mat\'{u}\v{s}-type.
\item If there exists $j\in N_n\backslash\{1,2,i\}$ not parallel with $\{i\}$ or a loop. Then $\{j\}$ is parallel with \{1\} or \{2\}( WLOG, we assume $\{j\}$ and $\{1\}$ are parallels.) or $\{1,2,i,j\}$ is $U_{2,4}$. For $\mathbf{h}=(a,b)\in F$, $\mathbf{h}=a\mathbf{r}_{M}+b\mathbf{r}_{U_{1,2}^n}$, Restricting $\mathbf{h}$ on $\{2,i,j\}$, we obtain $\mathbf{h}'=a\mathbf{r}_1+b\mathbf{r}_2$, where $\mathbf{r}_1$ and $\mathbf{r}_2$ are the rank functions of a matroid $U_{2,3}$ and $U_{1,1}^{3}$ on $\{2,i,j\}$, respectively. Thus, by Lemma \ref{u11}, if $\mathbf{h}$ is entropic, $a=\log v$ for positive integer $v$, that is, $F$ is Chen-Yeung-type.
\end{enumerate}
\end{enumerate}
\end{enumerate}
\end{proof}
\addtolength{\topmargin}{0.1in}
\subsection{All $(M,U_{1,n'}^n)$ are in one of the four types}\label{3c}
\begin{lemma}\textup{(Generalization of \cite[Lemma 3]{liu2023entropy})}\label{component}
	For a random vector $X_i,i\in N_{n-1}$, consider the (n-1)-partite graph $G=(V,E)$ with $V=\bigcup\limits_{i \in N_{n-1}}\mathcal{X}_i$ and $(x_i,x_j)\in E$ if and only if $p(x_i,x_j)>0$, $i,j\in N_{n-1}$, $i\neq j$. If $(X_i,i\in N_{n-1})$ satisfies the following infomation equalities,
	\begin{align}
		H(X_{i\cup K})+&H(X_{j\cup K})=H(X_K)+H(X_{ij\cup K}),\nonumber\\
		&i,j\in N_n,K \subseteq N_n\backslash\{i,j\} \nonumber
	\end{align}
	Each connected component of $G$ is a complete (n-1)-partite graph. Futhermore,
	if $p(x_1)= p(x_2)=...= p(x_{n-1})$ holds for any $p(x_{12...n-1}) > 0$, then the number of vertices in $\mathcal{X}_i, i=1,2,...,n-1$ are the same and the
	the probability mass of all of the vertices, the edges and the triangles are equal, respectively, in each connected component.
\end{lemma}
\begin{lemma}\label{u1n-1}
For any $U_{n-1,n}$ with $n\geq3$, $F=(U_{n-1,n},U_{1,n-1}^n)$, $\mathbf{h}=(a,b)\in F$ is entropic if and only if $a=\log v$ for positive integer $v$, that is, $F$ is Chen-Yeung-type.
\end{lemma}
\begin{proof}
If $\mathbf{h}\in F$ is entropic, its characterizing random vector($X_i,i\in N_n$) satisfies the following infomation equalities,
\begin{align}
H(X_{N_n})&=H(X_{N_{n}-i}),i\in N_n \label{eq:u1n-1.1} \\
H(X_{in})&=H(X_i)+H(X_n),i\neq n \label{eq:u1n-1.2} \\
H(X_{i\cup K})+H(X_{j\cup K})&=H(X_K)+H(X_{ij\cup K}).\label{eq:u1n-1.3}\nonumber \\ 
i,j\in N_n, K\subseteq N_n&\backslash\{i,j\},K\neq\{n\},|K|<n-2
\end{align}
Routine calculation leads to 
\begin{align}
H(X_{i_1j_1})=H(X_{i_2j_2}).\qquad \{i_1,j_1\}\neq\{i_2,j_2\}\label{eq:u1n-1.4}
\end{align}
According to the equalitiy \eqref{eq:u1n-1.1}-\eqref{eq:u1n-1.4}, routine calculation leads to
\begin{align}
H(X_{N_n})&=H(X_{N_n-1})=H(X_{N_n-\{12\}})+2H({X_1})=...\nonumber \\
&=H(X_1)+(n-2)H(X_n).\label{eq:u1n-1.5}
\end{align}
For $(x_i,i\in N_n)\in \mathcal{X}_{N_n}$, with $p(x_{N_n})>0$, above information equalities imply the probability mass function satisfies
\begin{align}
p(x_{N_n})=p(x_1)p(x_n)^{n-2}.\label{eq:u1n-1.6}
\end{align}
By equalities \eqref{eq:u1n-1.2},$X_1$ and $X_n$ are independent. For any $x'_n\in\mathcal{X}_n$ with $x'_n\neq x_n$,
\begin{align}
p(x_1,x'_n)=p(x_1)p(x'_n)>0.\label{eq:u1n-1.7}
\end{align}
As
\begin{align}
p(x_1,x'_n)=\sum_{x_2,x_3...x_{n-1}}p(x_1,x_2,x_3...x'_n),\label{eq:u1n-1.8}
\end{align}
there exists $x'_2\in \mathcal{X}_2$, $x'_3\in\mathcal{X}_3$, ..., $x'_{n-1}\in\mathcal{X}_{n-1}$ such that $p(x_1,x'_2,x'_3...x'_n)>0$. By the same argument, we have
\begin{align}
p(x_1,x'_2,x'_3...x'_n)=p(x_1)p(x'_n)^{n-2}.\label{eq:u1n-1.9}
\end{align}
For the $(n-1)$-partite graph $G=(V,E)$ with $V=\mathcal{X}_1\cup\mathcal{X}_2\cup\mathcal{X}_3...\cup\mathcal{X}_{n-1}$ and $(x_i,x_j)\in E$ if and only if $p(x_i,x_j)>0$, $i,j\in N_{n-1}$, $i\neq j$. By Lemma \ref{component}, the probablity mass of the $(n-1)$-side polygon are equal in each connected component. Thus, we can get $p(x_1,x_2,x_3,...x_n)=p(x_1,x'_2,x'_3,...x'_n)$. Combining equality\eqref{eq:u1n-1.6} and equality\eqref{eq:u1n-1.9}, we can get 
\begin{align}
p(x_n)=p(x'_n).\label{eq:u1n-1.10}
\end{align}
For any $x'_n\in \mathcal{X}_n$, which implies that $X_n$ is uniformly distributed on $\mathcal{X}_n$, and so $H(X_n)=\log v$ and $v=|\mathcal{X}_n|$.

The "only if" part can immediately implied by Lemma \ref{h1h2} and the fact that $a = \log v$ on the ray $U_{n-1,n}$, and whole ray $U_{1,n-1}^n$ are entropic.
\end{proof}
\begin{lemma}\textup{\cite[Lemma 6]{csirmaz2025exploring}\cite[Lemma 15.3]{yeung2008information}}\label{all}
For any extreme ray $P$, the 2-dimensional face $(P,U_{1,n'}^n)$ is all-entropic if
 \begin{enumerate}
 	\item $\mathbf{r}(P)=1$, or 
 	\item $\mathbf{r}(P)>1$ but each element that is a loop of $U_{1,n'}^n$ is also a loop of $P$.
 \end{enumerate}
\end{lemma}

We now define a matroid denoted by \( M_1^{123,145} \). This matroid is on a ground set of five elements, such that the restriction to the subset \(\{1,2,3\}\) is \( U_{2,3} \), the restriction to \(\{1,4,5\}\) is also \( U_{2,3} \), and the restriction to \(\{2,3,4,5\}\) is \( U_{3,4} \).

\begin{lemma}\label{M1}
For $F=(M_1^{123,145},U_{1,4}^{2345,5})$, $\mathbf{h}=(a,b)\in F$ is entropic if and only if $a=\log v$ for positive integer v, that is, $F$ is Chen-Yeung-type.
\end{lemma}
\begin{proof}
If $\mathbf{h}\in F$ is entropic, its characterizing random vector($X_i,i\in N_5$) satisfies the following infomation equalities,
\begin{align}
H(X_{N_5})&=H(X_{N_{5}-i}),i\in N_5\\
H(X_1)+H(X_i)&=H(X_1X_i), i\in N_5\\
H(X_{ik})+H(X_{jk})&=H(X_k)+H(X_{ijk}),\nonumber\\
i,j,k\in N_5,k\neq1, &\{ijk\}\neq\{123\},\{145\}\\ 
H(X_{i\cup K})+H(X_{j\cup K})&=H(X_K)+H(X_{ij\cup K}),\nonumber \\ 
i,j\in N_n, K\subseteq N_n&\backslash\{i,j\},|K|=2, \nonumber\\
K\neq\{24\},\{25\},\{34\},&\{35\},\{ij\}\neq\{23\},\{45\}\\
H(X_{i\cup K})+H(X_{j\cup K})&=H(X_K)+H(X_{ij\cup K}),\nonumber \\ 
i,j\in N_n,K\subseteq N_n\backslash\{i,j\}&,|K|=3,K\neq\{123\},\{145\}
\end{align}
For $(x_i,i\in N_n)\in \mathcal{X}_{N_n}$, with $p(x_{N_n})>0$, above information equalities imply the probability mass function satisfies
\begin{align}
p(x_{N_5})&=p(x_{N_{5}-i}),i\in N_5\\
p(x_1)p(x_i)&=p(x_1x_i), i\in N_5\\
p(x_{ik})p(x_{jk})&=p(x_k)p(x_{ijk}),\nonumber\\
i,j,k\in N_5,k\neq1, &\{ijk\}\neq\{123\},\{145\}\\ 
p(x_{i\cup K})p(x_{j\cup K})&=p(x_K)+p(x_{ij\cup K}),\nonumber \\ 
i,j\in N_n, K\subseteq N_n&\backslash\{i,j\},|K|=2, \nonumber\\
K\neq\{24\},\{25\},\{34\},&\{35\},\{ij\}\neq\{23\},\{45\}\\
p(x_{i\cup K})p(x_{j\cup K})&=p(x_K)p(x_{ij\cup K}),\nonumber \\ 
i,j\in N_n,K\subseteq N_n\backslash\{i,j\}&,|K|=3,K\neq\{123\},\{145\}
\end{align}
Routine calculation leads to 
\begin{align}
p(x_2=p(x_3)&=p(x_4)=p(x_5)\\
p(x_{N_{5}})&=p(x_{1234})=p(x_{134})\nonumber\\
&=p^2(x_1)p(x_3)
\end{align}
Note that $X_1$ and $X_3$ are independent. For any $x'_1\in\mathcal{X}_1$ with $x'_1\neq x_1$,
\begin{align}
p(x'_1,x'_3)=p(x'_1)p(x_3)>0.
\end{align}
As
\begin{align}
p(x'_1,x_3)=\sum_{x_2,x_4}p(x_{1234}),
\end{align}
there exists $x'_2\in \mathcal{X}_2$, $x'_4\in\mathcal{X}_4$ such that $p(x'_1,x'_2,x_3.,x'_4)>0$. By the same argument, we have
\begin{align}
p(x'_1,x'_2,x_3,x'_4)=p(x'_1)^2p(x_3).
\end{align}
For the tripartite graph $G=(V,E)$ with $V=\mathcal{X}_2\cup\mathcal{X}_3\cup\mathcal{X}_4$ and $(x_i,x_j)\in E$ if and only if $p(x_i,x_j)>0$, $i,j\in {234}$, $i\neq j$. By Lemma \ref{component}, the probablity mass of the triangles are equal in each connected component. Thus, we can get $p(x_1,x_2,x_3,x_4)=p(x'_1,x'_2,x_3,x'_4)$. Then, we can get 
\begin{align}
p(x_1)=p(x'_1).
\end{align}
For any $x'_1\in \mathcal{X}_1$, which implies that $X_1$ is uniformly distributed on $\mathcal{X}_1$, and so $H(X_1)=\log v$ and $v=|\mathcal{X}_1|$.

The "only if" part can immediately implied by Lemma \ref{h1h2} and the fact that $a = \log v$ on the ray $M_1^{123,145}$, and whole ray $U_{1,4}^{2345,5}$ are entropic.
\end{proof}
\addtolength{\topmargin}{0.01in}
For a polymatroid $P=(\mathbf{h},N_n)$ and a subset $S\subseteq N_n$, we say $P^S=(\mathbf{h}_S,S)$ with $\mathbf{h}_S=\mathbf{h}$ for all $A\subseteq S$ a restriction of $P$ on $S$. It can be seen that $P^S$ is a polymatroid in $\Gamma_S$. Now for a 2-dimensional face $(P_1,P_2)$ of $\Gamma_n$, we call $F_S=(P_1^S,P_2^S)$ is a $\mathit{restricted}$ $\mathit{face}$ of $(P_1,P_2)$ on $S$ if $(P_1^S,P_2^S)$ is a 2-dimensional face of $\Gamma_S$. 
\begin{theorem}\label{11}
Let $M$ be a matroid with $C$ the family of circuits, and it is in an extreme ray of $\Gamma_n$. If $F=(M,U_{1,n'}^n)$ is a 2-dimensional face of $\Gamma_n$, then the face $(M,U_{1,n'}^n)$ is
\begin{enumerate}
\item all-entropic if $M$ is rank-1 or restricted faces $F_C$ are all-entropic for $C\in\mathcal{C}$,
\item Mat\'{u}\v{s}-type if  
\begin{itemize}
	\item there exists $C\in\mathcal{C}$ such that $F_C$ is Mat\'{u}\v{s}-type; and
	\item restricted faces $F_{C'}$, for $C'\neq C\in \mathcal{C}$ are all entropic-type; and
	\item any loop of $M$ is a loop of $U_{1,n'}^n$;and
	\item $\mathbf{r}(M)=2$.
\end{itemize}  
\item Chen-Yeung-type if
\begin{enumerate} 
\item there exists $C\in\mathcal{C}$ such that $F_C$ is Chen-Yeung-type or,
\item \begin{itemize}
	\item there exists $C\in\mathcal{C}$ such that $F_C$ is Mat\'{u}\v{s}-type; and
	\item restricted faces $F_{C'}$, for $C'\neq C\in \mathcal{C}$ are all entropic-type; and
	\item there exists a loop of $M$ not a loop of $U_{1,n'}^n$ or any loop of $M$ is a loop of $U_{1,n'}^n$ but $\mathbf{r}(M)>2$.
\end{itemize} 
\end{enumerate}
\end{enumerate}
\end{theorem}
\begin{proof}
If a matroid $M$ is non-entropic, then $(M,U_{1,n'}^n)$ is non-entropic. Thus, the matroid $M$ we discuss in the following proof is entropic.

If $M$ is rank-$1$, it is all-entropic by Lemma \ref{all}. For a matroid $M$ with $\mathbf{r}(M)\geq 2$ in an extreme ray and circuit $C\in \mathcal{C}$, we classify the restricted faces $F_C$ of $(M,U_{1,n'}^n)$ into the following three types,
\begin{enumerate}
\item all 2-dimensional faces $(U_{n-1,n},U_{1,n})$ for integer $n\geq 2$. By Lemma \ref{all}, they are all entropic.
\item $(U_{2,3},U_{1,2})$ is Mat\'{u}\v{s}-type \cite{matus2005piecewise}.
\item $(U_{2,3},U_{1,1})$, $(U_{n-1,n},U_{1,k})$ for positive integer $n,k$ with  $n\geq4$, $k<n$ and all uniform matroids. By Lemma \ref{u11}, \ref{u12} and \ref{u1n-1}, they are Chen-Yeung type.

\end{enumerate}
For 2-dimensional faces $(M,U_{1,n'}^n)$, we classify them into three distinct types by the restricted faces $F_C$ we classify above.
\begin{enumerate}

\item If all restricted faces $F_C$ are all-entropic, we can easliy see each element is a loop in $U_{1,n'}^n$ is also a loop in $M$. Thus, the 2-dimensional face $(M,U_{1,n'}^n)$ is all entropic by Lemma \ref{all}.

\item If there exists one circuit $C\in\mathcal{C}$ such that the restricted face $F_C$ is Chen-Yeung-type, the 2-dimensional face is Chen-Yeung-type. For $\mathbf{h}=a\mathbf{r}_M+b\mathbf{r}_{U_{1,n'}^n}$, restricting $\mathbf{h}$ on $C$, we can obtain $\mathbf{h_C}=a\mathbf{r}_1+b\mathbf{r}_2$. Then $\{(a,b): a=\log v,b\geq 0, v\in\mathbb{Z}\}$ forms an outer bound on $F^*$ by the fact that the restricted face $F_C$ is Chen-Yeung-type. The inner bound $\{(a,b):a=\log v,b\geq 0, v\in \chi_M$\} is immediately implied by Lemma \ref{h1h2} and $\chi_M$. Hence, $F$ is Chen-Yeung-type. 

\item If there exists one restricted face of Mat\'{u}\v{s}-type and all other restricted faces are all-entropic, and $F_C$ is Mat\'{u}\v{s}-type for $C\in\mathcal{C}$. Then $n'=n-1$ and $(M,U_{1,n-1}^n)$ can be classified into Chen-Yeung-type or Mat\'{u}\v{s} type according to whether there is a loop of $M$ that is not a loop of $U_{1,n-1}^n$ and whether $\mathbf{r}(M)=2$.
\begin{enumerate}

\item We first prove that if there exists a loop $i\in N_n$ of $M$, which is not a loop of $U_{1,n-1}^n$, then $(M,U_{1,n-1}^n)$ is Chen-Yeung-type. We assume $C=\{1,2,n\}$ be a circuit of $M$ and $n$ is a loop of $U_{1,n-1}^n$. For $\mathbf{h}=a\mathbf{r}_M+b\mathbf{r}_{U_{1,n'}^n}$, restricting $\mathbf{h}$ on $\{1,2,n,i\}$, we can obtain $\mathbf{h}'=a\mathbf{r}'_1+b\mathbf{r}'_2$, where $\mathbf{r}'_i$, $i=1,2$ are the rank functions of $U_{2,3}^{12n,4}$ and $U_{1,3}^{12i,4}$ on $\{1,2,n,i\}$, respectively. Thus, by the fact $(U_{2,3}^{12n,4},U_{1,3}^{12i,4})$ is Chen-Yeung-type\cite[Theorem 4]{liu2023entropy}, $\{(a,b): a=\log v,b\geq 0, v\in\mathbb{Z}\}$ forms an outer bound on $F^*$. The inner bound $\{(a,b):a=\log v,b\geq 0, v\in \chi_M$\} is immediately implied by Lemma \ref{h1h2} and $\chi_M$. Hence, $F$ is Chen-Yeung-type. 

\item If any loop of $M$ is a loop of $U_{1,n-1}^n$ and $\mathbf{r}(M)=2$, we prove $(M,U_{1,n-1}^n)$ is Mat\'{u}\v{s}-type. Restricting $\mathbf{h}$ on $C$, we can obtain $\mathbf{h}_C=a\mathbf{r}'_1+b\mathbf{r}'_2$. The outer bound on $(M,U_{1,n-1}^n)$ is $a+b=\log \lceil e^a \rceil$ by the fact that the outer bound of $(U_{2,3},U_{1,2}^3)$ is $a+b=\log \lceil e^a \rceil$. Let $\mathbf{Y}_n$ be a random vector whose entropy function is $\log v\cdot \mathbf{r}_M$, $v\in\chi_M$. As $\mathbf{r}_M= 2$, according to \cite[Theorem 1]{chen2021matroidal}, $\mathbf{Y}_n$ is uniformly distributed on its support $S$. Let $\mathbf{X}_n=(X_i:i\in N_n)$ be a random vector distributed on $S$ with $H(X_n)=a$, and $p(x_{N_n})=\frac{p(x_n)}{v}$, for any $X_N\in S$. It can be checked that the entropy function of such constrcuted $\mathbf{X}_n$ is $a+b=\log v$ for $v\in \chi_M$. Then by Lemma \ref{h1h2} and the fact that the whole ray $U_{1,n'}^n$ is entropic, the inner bound of $(M,U_{1,n'}^n)$ is $a+b>\log v$ and $\log (v-1)<a\leq\log v$ for integer $v\in \chi_M$. Hence, $F$ is Mat\'{u}\v{s}-type.

\item If any loop of $M$ is a loop of $U_{1,n-1}^n$ and $\mathbf{r}(M)>2$, we prove $(M,U_{1,n-1}^n)$ is Chen-Yeung-type. There exists $C_1\in\mathcal{C}$ such that $F_{C_1}$ is Mat\'{u}\v{s}-type and we may assume that $1\in{C_1}$ and 1 is a loop in $U_{1,k}$ but not a loop in $M$. Since $\mathbf{r}(M)>2$, by Lemma \ref{u11}, \ref{u12}, \ref{u1n-1}, \ref{all}, any $F_{C'}$ is ${U_{n'-1,n'},U_{1,n'}}$ for $C'\in \mathcal{C}$ and $\mathbf{r}_M(C')=n'>2$, which implies 1 does not appear in any $C'\in\mathcal{C}$ with $\mathbf{r}_M(C)>2$. Furthermore, $M$ is a connected matroid, thus any $F_{C_2}$ is Mat\'{u}\v{s}-type for  $C_2\in \mathcal{C}$ and $1\in C_2$. Based on $\mathbf{r}(M)>2$, we can deduce that there exist $C_3, C_4\in\mathcal{C}$ (both contain 1) such that $|C_3\cup C_4|=5$ and $\mathbf{r}(C_3\cup C_4)=3$. Then we can easily find the restricted face $F_{C_3\cup C_4}$ is $(M_1^{123,145},U_{1,4}^{2345,5})$, then $\{(a,b): a=\log v,b\geq 0, v\in\mathbb{Z}\}$ forms an outer bound on $F^*$ by the fact that the restricted face $F_{C_3\cup C_4}$ is Chen-Yeung-type by Lemma \ref{M1}. The inner bound $\{(a,b):a=\log v,b\geq 0, v\in \chi_M$\} is immediately implied by Lemma \ref{h1h2} and $\chi_M$. Hence, $F$ is Chen-Yeung-type. 
\end{enumerate}
\end{enumerate}
\end{proof}
Theorem \ref{11} immediately implies the following theorem.
\begin{theorem}
	For matroid $M$ in an extreme ray, if $(M,U_{1,n'}^n)$ is a 2-dimensional face of $\Gamma_n$, then the face $(M,U_{1,n'}^n)$ is one of the four types in Definition \ref{d1}.
	\end{theorem}
\bibliographystyle{IEEEtran}
\bibliography{refe.bib}

\end{document}